\documentclass{llncs}

\usepackage{makeidx}

\usepackage[utf8]{inputenc}
\usepackage{amsmath, amssymb, verbatim, bbm,amsfonts, hyperref, mathrsfs}
\usepackage{amsthm2}
\usepackage[margin=1in]{geometry}

\usepackage{authblk}
\usepackage{amssymb,latexsym}

\usepackage{caption}
\usepackage{subcaption}
\usepackage{color}
\usepackage[pdftex]{graphicx}

\makeatletter
\newtheorem*{rep@theorem}{\rep@title}
\newcommand{\newreptheorem}[2]{%
	\newenvironment{rep#1}[1]{%
		\def\rep@title{#2 \ref{##1}}%
		\begin{rep@theorem}}%
		{\end{rep@theorem}}}
\makeatother

\newreptheorem{theorem}{Theorem}
\newreptheorem{lemma}{Lemma}
\newreptheorem{proposition}{Proposition}

\newcommand{\altketbra}[1]{\ketbra{#1}{#1}}

\newcommand{\FQ}{\mathbf{\mathsf{\mathbbm{F}_Q}}}
\newcommand{\kb}[1]{\altketbra{#1}}
\newcommand{\Comm}{\emph{Comm}}

\newcommand{\Open}{\emph{Open}}
\newcommand{\Str}{\emph{Str}} 
\newcommand{\CHSH}{\mathbf{ \mathsf{CHSH}}}
\newcommand{\CHSHQ}{\mathbf{ \mathsf{CHSH_Q}}}
\newcommand{\F}{\mathbb{F}}
\newcommand{\E}{\mathbb{E}}
\newcommand{\etal}{\emph{et al.~}}
\newcommand{\I}{\mathbb{I}}

\newcommand{\ie}{\textit{i.e. }}

\newcommand{\HamiltonianCycle}{\textsc{Hamiltonian Cycle }}

\newcommand{\COMMENT}[1]{}
\newcommand{\spa}[1]{\mathcal{#1}}
\newcommand{\ket}[1]{|#1\rangle}

\newcommand{\ketbra}[2]{|#1\rangle\langle#2|}
\def\01{\{0,1\}}
\newcommand{\braket}[2]{\langle{#1}|{#2}\rangle} 
\def\01{\{0,1\}}

\newcommand{\triple}[3]{\langle{#1}|{#2}|{#3}\rangle}
\newcommand{\eps}{\varepsilon}

\newcommand{\zo}{\{0,1\}}

\begin{document}
\title{Relativistic (or $2$-prover $1$-round) zero-knowledge protocol for $\mathsf{NP}$ secure against quantum adversaries}
\author{André Chailloux, Anthony Leverrier}
\authorrunning{André Chailloux, Anthony Leverrier}
\institute{Inria, Paris \\
\email{andre.chailloux@inria.fr $\quad$ anthony.leverrier@inria.fr}}
\maketitle

\begin{abstract}
	In this paper, we show that the zero-knowledge construction for \HamiltonianCycle remains secure against quantum adversaries in the relativistic setting.  
	Our main technical contribution is a tool for studying the action of consecutive measurements on a quantum state which in turn gives upper bounds on the value of some entangled games. This allows us to prove the security of our protocol against quantum adversaries. We also prove security bounds for the (single-round) relativistic string commitment and bit commitment in parallel against quantum adversaries.
	As an additional consequence of our result, we answer an open question from \cite{Unr12} and show tight bounds on the quantum knowledge error of some $\Sigma$-protocols.
\end{abstract}

\section{Introduction}
\subsection{Context}
The goal of relativistic cryptography is to exploit the no superluminal signaling (NSS) principle in order to perform various cryptographic tasks. NSS states that no information carrier can travel faster than the speed of light. Note that this principle is closely related to the non-signaling principle that says that a local action performed in a laboratory cannot have an \emph{immediate} influence outside of the lab. NSS is more precise since it gives an upper bound on the speed at which such an influence can propagate. 
Apart from this physical principle, we want to ensure {information-theoretic} security meaning that the schemes proposed cannot be attacked by any classical (or quantum) computers, even with unlimited computing power.

The idea of using the NSS principle for cryptographic protocols originated in a pioneering work by Kent in 1999 \cite{Kent99} as a way to physically enforce a no-communication constraint between the different agents of one party (the idea of splitting up a party into several agents dates back to \cite{BGK88}, but without any explicit implementation proposal). The original goal of Kent was to bypass the no-go theorems for quantum bit-commitment \cite{May97,LC97}. More recently, quantum relativistic bit commitment protocols were developed where the parties exchange quantum systems, with the hope that combining the NSS principle together with quantum theory would lead to more secure (but less practical) protocols \cite{Kent11,Kent12,KTH13}. In particular, the protocol \cite{Kent12} was implemented in ~\cite{LKB13}.  We note that the scope of relativistic cryptography is not limited to bit commitment. For instance, there was recently some interest (sparked again by Kent) for position-verification protocols \cite{KMS10,LL11,Unr14} but contrary to the case of bit commitment, it was shown that secure position-verification is impossible both in the classical and the quantum settings \cite{CGM09,BCF14}.

The original idea of \cite{BGK88} was recently revisited by Cr\'epeau \etal in \cite{CSST11} (see also \cite{sim07}). Based on this work, Lunghi \textit{et al.}~devised a bit commitment protocol involving only four agents, two for Alice and two for Bob \cite{LKB+15}. Their protocol is secure against quantum adversaries and a multi-round variant, with longer duration time, was shown to be secure against classical adversaries \cite{LKB+15,CCL15,FF15}. While those protocols only seemed of theoretical interest at first, recent implementations have convincingly demonstrated that the required timing and location constraints can be efficiently enforced. In \cite{VMH16}, the authors performed a $24$-hour-long bit commitment with the pairs of agents standing $8$km apart.

The security analysis against quantum adversaries of \cite{LKB+15} and against classical adversaries of \cite{CCL15} relies on the study of variants of $\CHSH$ games where the inputs and outputs belong to the field $\FQ$ for some large prime power $Q$, instead of $\{0,1\}$ for the usual $\CHSH$ game. In many cases, the (quantum) security of a relativistic protocol can be derived from the value of an (entangled) $2$-player game. Because the relativistic constraint essentially boils down to $2$ non-communicating provers, a relativistic protocol can also be seen as a $2$-prover interactive protocol.

\begin{center} --- \end{center}

The above results are promising for relativistic cryptography but very limited in scope. Indeed, bit commitment schemes are used as parts of larger cryptosystems. The only study of the composability of the $\FQ$ bit commitment scheme was done in \cite{FF15} but mainly with itself, in order to increase the commit time. There has not been any proposition to use this scheme for a more general purpose. 

One natural application of bit commitment are zero-knowledge protocols. With such a protocol, a prover wishes to convince a verifier that a given statement is true without revealing any extra information. A zero-knowledge protocol is already a more advanced cryptographic primitive and has more direct applications such as identification schemes \cite{GMR89} for instance. Here, we will consider the zero-knowledge construction for \HamiltonianCycle, which is an $\mathsf{NP}$ complete problem. The prover will convince the verifier that a given graph $G =(V,E)$ has a Hamiltonian cycle, \ie a cycle going through each vertex exactly once, without revealing any information, in particular no information about this cycle.
Since \HamiltonianCycle is $\mathsf{NP}$ complete, a zero-knowledge protocol for this problem can be used to obtain a zero-knowledge protocol for arbitrary $\mathsf{NP}$ problems.

There is a known zero-knowledge protocol for \HamiltonianCycle \ using bit commitment first presented by Blum \cite{Blu86} which we recall now. \\ \\

\fbox{
\begin{minipage}{0.93\textwidth}
\begin{center}
{ Zero-knowledge protocol for \HamiltonianCycle\ using bit commitment } \end{center}
\begin{enumerate}
	\item The prover picks a random permutation $\Pi : V \rightarrow V$. He commits to each of the bits of the adjacency matrix $M_{\Pi(G)}$ of $\Pi(G)$.
	\item The verifier sends a random bit (called the challenge) $chall \in \zo$ to the prover.
	\item
	\begin{itemize}
		\item If $chall = 0$, the prover decommits to all the elements of $M_{\Pi(G)}$, and reveals $\Pi$. 
		\item If $chall =1$, he reveals only the bits (of value $1$) of the adjacency matrix that correspond to a Hamiltonian cycle $\mathcal{C}'$ of $\Pi(G)$.
	\end{itemize}
	\item The verifier checks that these decommitments are valid and correspond, for $chall = 0$ to $M_{\Pi(G)}$ and, for $chall=1$, to a Hamiltonian cycle.
\end{enumerate}
\end{minipage} 
} 
$ \ $ \\

It is natural to combine this zero-knowledge protocol with the $\FQ$ relativistic bit commitment protocol mentioned above. The (single-round) $\FQ$ relativistic bit commitment protocol is secure against quantum adversaries but it doesn't directly imply that the zero-knowledge protocol remains secure. Indeed, the security definition considered for the bit commitment is fairly weak and composes poorly with other protocols. The soundness of the protocol against entangled provers will be reduced to a $2$-player entangled game. Proving zero-knowledge against a quantum verifier can sometimes be complicated because of the presence of a quantum auxiliary input. In this case, however, due to properties of the relativistic $\FQ$ bit commitment, we will not need any rewinding from our simulator and the simulation will actually be rather simple.

\begin{center} --- \end{center}

The goal of this paper is to show that it is indeed possible to plug in the $\FQ$ relativistic bit commitment protocol into Blum's zero-knowledge protocol for \HamiltonianCycle. This widens the possible applications for relativistic cryptography and will encourage further implementations.

The main contribution of this paper is a technical analysis involving successive measurements on a quantum system. Indeed, to prove that the above scheme is secure against quantum adversaries, we use the fact that an adversary who can answer both challenges at the same time can guess the value of a string on which he has no information, due to non-signaling. This naturally involves consecutive measurements on a quantum system, and leads us to analyze how the first measurement disturbs the system before the second measurement.

\subsection{Relativistic zero-knowledge protocol for \HamiltonianCycle}
Here, we show how the final protocol will look like and where exactly we rely on the physical NSS principle. A similar protocol, with two non communicating provers, appeared in \cite{LS95}. The final protocol is the following: \\

\fbox{
	\begin{minipage}{0.93\textwidth}
		\begin{center}	Relativistic zero knowledge protocol for \HamiltonianCycle\  \end{center} 
		
		\noindent \textbf{Input} --- The provers and the verifiers are given a graph $G = (V,E)$. \\
		\noindent \textbf{Auxiliary Input} --- The provers $P_1$ and $P_2$ know a Hamiltonian cycle $\mathcal{C}$ of $G$. \\
		\noindent \textbf{Preprocessing} --- $P_1$ and $P_2$ agree beforehand on a random permutation $\Pi : V \rightarrow V$ and on an $n \times n$ matrix $A \in \mathcal{M}_n^{\F_Q}$ where each element of $A$ is chosen uniformly at random in $\F_Q$. \\
		\noindent \textbf{Protocol} --- 
		\begin{enumerate}
			\item Commitment to each bit of $M_{\Pi(G)}$ : $V_1$ sends a matrix $B \in \mathcal{M}_n^{\F_Q}$ where each element of $B$ is chosen uniformly at random in $\F_Q$. $P_1$ outputs the matrix $Y \in \mathcal{M}_n^{\F_Q}$ such that $\forall i,j \in [n], \ Y_{i,j} = A_{i,j} + (B_{i,j} * (M_{\Pi(G)})_{i,j})$.
			\item The verifier sends a random bit (called the challenge) $chall \in \zo$ to the prover.
			\item
			\begin{itemize}
				\item If $chall = 0$, $P_2$ decommits to all the elements of $M_{\Pi(G)}$, \textit{i.e.} he sends all the elements of $A$ to $V_2$  and reveals $\Pi$. 
				\item If $chall =1$, $P_2$ reveals only the bits (of value $1$) of the adjacency matrix that correspond to a Hamiltonian cycle $\mathcal{C}'$ of $\Pi(G)$, \textit{i.e.} for all edges $(u,v)$ of $\mathcal{C}'$, he sends $A_{u,v}$ as well as $\mathcal{C}'$.
			\end{itemize}
			\item The verifier checks that those decommitments are valid and correspond to what the provers have declared. He also checks that the timing constraint of the bit commitment is satisfied. This means that
			\begin{itemize}
				\item if $chall = 0$, the prover's opening $A$ must satisfy $\forall i,j \in [n], \ Y_{i,j} = A_{i,j} + (B_{i,j} * (M_{\Pi(G)})_{i,j})$.
				\item if $chall = 1$, the prover's opening $A$ must satisfy $\forall (u,v) \in \mathcal{C}', \ Y_{u,v} = A_{u,v} + B_{u,v}$.
			\end{itemize}
		\end{enumerate} 
	\end{minipage} 
} $ \ $ \\

The above protocol is obtained by plugging in the $\F_Q$ relativistic bit commitment protocol into Blum's zero-knowledge protocol for \HamiltonianCycle. We discuss more the setting in Sections \ref{Section:RelativisticBC} and \ref{Section:RelativisticZK}. We just want here to briefly present in which way we use the no superluminal signaling condition in this protocol.

In order for the protocol to be secure, we require the following:
\begin{enumerate}
	\item Both the prover and the verifier are split into $2$ agents, respectively $P_1,P_2$ and $V_1,V_2$.
	\item $V_1$ and $V_2$ are far apart (we discuss this later).
	\item The opening phase (steps $2$ and $3$) must be performed as soon as the commit phase (step $1$) is completed.
\end{enumerate}
Constraints $(2)$ and $(3)$ are here to enforce that during step $3$ of the protocol, \emph{the message $P_2$ sends to $V_2$ does not depend on the matrix $B$ sent by $V_1$}. Because information travels at most at light speed, by synchronizing the steps well enough, the verifiers can enforce this condition. For instance, it is sufficient to check that $V_2$ receives the message from $P_2$ before the information on $B$ sent by $V_1$ had time to reach $V_2$. If this is the case, then it guarantees that $P_2$'s answer to the challenge cannot depend on the value of the matrix $B$, since otherwise it would violate the NSS. 
An important consequence is that we do not require the verifiers to know anything about the spatial locations of the provers: it is sufficient for the verifiers to know their own relative position. 

As said before, there were already several experiments made that showed how to achieve the above constraints. The most notable one \cite{VMH16} succeeded in performing the above bit commitment protocol by having $V_1$ and $V_2$ being $8km$ apart, which shows it can be achievable in real life conditions.

In summary, the main contribution of the paper is to prove the security of the above protocol for \HamiltonianCycle against quantum adversaries. The main challenge is to prove the soundness property, \textit{i.e.}~security against a cheating prover on an input which \emph{does not} contain a Hamiltonian cycle. Here, we have two dishonest provers $P_1$ and $P_2$ that want to pass the protocol even though the input graph does not contain a Hamiltonian cycle.
A cheating prover that would be able to answer simultaneously to both challenges could break the underlying string commitment scheme, which is a consequence of the special soundness property of the scheme.

To prove the security of the above protocol against quantum adversaries, we will, from a cheating strategy, construct a strategy that will successfully answer both challenges by consecutively applying the cheating strategy for each challenge, which is expressed by our consecutive measurement theorem (see Theorem \ref{Theorem:Multi} below). We can also view this cheating scenario as a $2$-player entangled game and we will show how in general our theorem regarding quantum consecutive measurements can be translated into a bound on the entangled value of $2$-player games.

\subsection{Consecutive measurements}
Our main technical contribution is expressed by the following theorem
\begin{theorem}\label{Theorem:Multi}
	Consider $n$ projectors $P_1,\dots,P_n$ such that for each $i$, we can write $P_i := \sum_{s = 1}^S P_i^s$ where the $\{P_i^s\}_s$ are orthogonal projectors for each $i$, \textit{i.e.}~for each $i$ and $s,s'$, we have $P^s_i P^{s'}_i = \delta_{s,s'} P^s_i$. Let $\sigma$ be any quantum state, let $V := \frac{1}{n} \sum_{i=1}^n tr(P_i  \sigma)$, and let \\ $E := \frac{1}{n(n-1)} \sum_{i,j\neq i} \sum_{s,s' = 1}^S tr(P^{s'}_j P^s_i \sigma P^s_i P^{s'}_j)$. Then it holds that $E \ge \frac{1}{64S}\big(V - \frac{1}{n}\big)^3$.
\end{theorem}

Such a statement can be fairly easily transposed to the context of games: see Proposition \ref{Proposition:Quantum_coup_bound} below. 
This theorem can be seen as a generalization of the \textit{gentle measurement} lemma \cite{Win99}, which is similar to the above with $n = 2$ and $S = 1$. The case of $n=2$ can be seen as a worst-case consecutive measurement theorem: how much can the first measurement disturb the measured state before the second measurement? However, for larger values of $n$, this shows that when we pick $2$ measurements out of $n$, the disturbance is much smaller, as shown by the dependence of the lower bound in $n$. Our theorem also improves on known results since it deals with larger values of $S$.

Interestingly, this kind of statement has already appeared previously in a paper by Unruh \cite{Unr12}, who studied quantum sigma protocols and in particular quantum proofs of knowledge. He showed the following

\begin{theorem}[\cite{Unr12}]\label{Theorem:Unruh}
	Consider $n$ projectors $P_1,\dots,P_n$ and an arbitrary quantum state ${\sigma}$. Let $V := \frac{1}{n} \sum_{i=1}^n tr(P_i \sigma )$, and let $E := \frac{1}{n(n-1)} \sum_{i,j\neq i} tr(P_j P_i \sigma P_i P_j)$.
	If $V \ge \frac{1}{\sqrt{n}}$ then $E \ge V(V^2 - \frac{1}{n})$.
\end{theorem}

Let us now compare our main theorem to Unruh's one, for the case of $S = 1$ where they are comparable. If $V \gg \frac{1}{\sqrt{n}}$ then both bounds give essentially the same bound $E \ge \Omega(V^3)$ which will translate into the relation $\omega^*(G_{coup}) \ge \Omega(\omega^*(G)^3)$ for the entangled values of a game $G$ and its coupled version $G_{coup}$ (see below). However, Theorem \ref{Theorem:Unruh} is only valid when $V \ge \frac{1}{\sqrt{n}}$ while Theorem \ref{Theorem:Multi} works for any $V \ge \frac{1}{n}$. Moreover, Theorem \ref{Theorem:Multi} is tight in its extremal point in the sense that there exist a quantum state and $n$ projectors such that $V = \frac{1}{n}$ and $E = 0$, as can be seen by considering for example $\sigma = \kb{\phi}$ with $\ket{\phi} := \frac{1}{\sqrt{n}} \sum_{i} \ket{i}$ and $P_i = \kb{i}$.

\begin{center} --- \end{center}

A natural application of our consecutive measurement theorem is to bound the value of some entangled games. The phrasing in terms of nonlocal games is sometimes more comfortable to use. In this paper, our security proofs will usually reduce to bounding the entangled value of such game, that is the maximum winning probability for a pair of players allowed to share arbitrary entangled states as a resource. For any game $G$ on the uniform distribution (meaning that the inputs of the game are drawn independently from the uniform distribution), we define the game $G_{coup}$ (consisting of a certain \emph{couple} of instances of $G$) as follows:
\begin{itemize}
	\item In $G$, Alice and Bob respectively receive $x$ and $y$ taken from the uniform distribution on the sets $I_A$ and $I_B$, respectively, and output $a$ and $b$ such that $V(a,b|x,y) = 1$ for some valuation function $V$ specified by $G$.
 \item In $G_{coup}$, Alice receives a random $x$ as in $G$ and Bob receives a pair of distinct random inputs $(y,y')$. Alice outputs $a$ and Bob outputs a pair $(b,b')$. They win the game if $V(a,b|x,y) = 1$ and $V(a,b'|x,y') = 1$, that is, if they win both instances of the game $G$, but for the same input/output pair of Alice.
\end{itemize}

In many cases, upper bounding the value of $G_{coup}$ will follow directly from a non-signaling argument of the form: ``If the players are able to win $G_{coup}$ with probability $p$ then Bob can learn some (or all) bits of $x$ with probability $p$ and no-signaling implies that $p \leq 1/|I_A|$''. What is left to do is to relate the entangled values of both games, $\omega^*(G)$ and $\omega^*(G_{coup})$. To do this, we construct the following strategy for $G_{coup}$: Alice follows the same strategy as for $G$; on inputs $(y,y')$, Bob performs the same strategy (measurement) as for $G$ on input $y$ to get output $b$ and then on input $y'$ to get $b'$. Note here that the non trivial part is that Bob's second measurement is applied on the post-measurement state resulting from his first measurement. Because we are in the quantum setting, this first measurement will generally perturb the state shared by Alice and Bob, which makes it non trivial to relate the success probability of this strategy for $G_{coup}$ with the entangled value $\omega^*(G)$ of the original game $G$.

A similar construction of \emph{squared games} was introduced in \cite{DS14,DSV15} to study projective classical and entangled games. There, the input $x$ is not revealed to the players but they receive respectively $y$ and $y'$ and output $b$ and $b'$. They win if there exists $a$ such that $V(a,b|x,y) = V(a,b'|x,y')$ = 1. It would be interesting to see the similarities and differences between those two approaches.

We show the following.
\begin{proposition}\label{Proposition:Quantum_coup_bound}
	For any game $G$ on the uniform distribution which is $S$-projective, we have $\omega^*(G_{coup}) \ge \frac{1}{S \cdot 64} \cdot (\omega^*(G) - \frac{1}{n})^3$ where $n$ is dimension of Bob's input.
\end{proposition} 
A game $G$ is said to be $S$-projective if for all $x,y,a$, there are at most $S$ possible outputs for Bob that allow them to win the game, \textit{i.e.}~$\max_{x,y,a} | \{ b\: :\: V(a,b|x,y)=1\}|\leq S$. 

In order to prove this statement, we need to analyze the strategy that we presented above. As already mentioned, the main difficulty is that the first measurement from Bob will modify the common shared state and therefore we cannot directly bound the probabilities related to the second measurement. One way of analyzing these consecutive measurements would be to use a kind of gentle measurement lemma but unfortunately, this would only work when the winning probability $\omega^*(G)$ is close to 1, which isn't the case for the games we consider. 

Fortunately, Theorem \ref{Theorem:Multi} is tailored for this kind of applications and can be used directly to prove the above proposition. We can notice the exact transposition of the parameters of Theorem \ref{Theorem:Multi} to Proposition \ref{Proposition:Quantum_coup_bound}.

\subsection{Applications of the bound}
\begin{enumerate}
	\item First, we prove that the extensions of the $\FQ$ bit commitment to string commitment and its parallel repetition remain secure against quantum adversaries with using the sum-binding definition. This is a direct consequence on upper bounds on the entangled value of $\CHSH$ variants, like the $\CHSHQ(P)$ game introduced in \cite{CCL15}. 
	\item We show that the presented relativistic zero-knowledge protocol for \HamiltonianCycle\ is secure against quantum adversaries. This also implies a $2$-prover $1$-round zero-knowledge protocol for \HamiltonianCycle\ also secure against quantum adversaries. 
	\item Finally, as a direct corollary of our consecutive measurement claim, we answer an open question from Unruh regarding quantum proofs of knowledge \cite{Unr12}. We show tight bounds on the quantum knowledge error of a $\Sigma$-protocol with  strict and special soundness as function of the challenge size, matching the classical bound. We will not discuss in detail this result as it just requires to plug our bound in the proof of \cite{Unr12} and is a bit beyond the scope of this paper. However, this shows that our results are useful beyond just the study of relativistic protocols or entangled games. 
\end{enumerate}
The last point shows that our bound could find even more applications when considering security against quantum adversaries. Indeed, when studying cryptographic protocols, for instance $\Sigma$-protocols, a notion that often appears is \emph{special soundness} which roughly states that an attacker shouldn't be able to simultaneously answer successfully to $2$ verifier's challenges. The relativistic zero-knowledge protocol we study is one example of this and Unruh's quantum proofs of knowledge setting is another one but there are more where our theorem could be useful.

\subsection*{Organisation of the paper}
In Section \ref{Section:Consecutive measurement propositions}, we prove our main consecutive measurement theorem. In Section \ref{Section:EntangledGames}, we show how to use this bound for proving upper bounds on the entangled value of nonlocal games. In Section \ref{Section:RelativisticBC}, we present in more detail the relativistic model and the $\FQ$ relativistic bit commitment protocol. Finally, in Section \ref{Section:RelativisticZK}, we describe the protocol obtained by plugging this bit commitment into Blum's zero-knowledge protocol for \HamiltonianCycle and we prove that it remains secure, even against quantum adversaries.

\section{Consecutive measurement theorems} \label{Section:Consecutive measurement propositions}
We first present some useful lemmata in the preliminaries. Then, we dive in directly in the proof of our consecutive measurements theorems.
\subsection{Preliminaries}
\begin{lemma} \label{Lemma:Projectors1}
	Let $\ket{\phi}$ a quantum pure state, $P \le \I$ a projector acting on $\ket{\phi}$ and $\ket{\psi} := \frac{P(\ket{\phi})}{||P(\ket{\phi})||}$. We have $|\braket{\phi}{\psi}|^2 = ||P(\ket{\phi})||^2 = tr(P \kb{\phi})$.
\end{lemma}
\begin{proof}
	We write $\ket{\phi} = P(\ket{\phi}) + (\I - P) (\ket{\phi}) = ||P(\ket{\phi})|| \ \ket{\psi} + (\I - P) (\ket{\phi})$. By noticing that $\triple{\psi}{\I - P}{\phi} = 0$, we get $|\braket{\phi}{\psi}|^2 = ||P(\ket{\phi})||^2 = tr(P \kb{\phi})$.
\end{proof}
\begin{lemma} \label{Lemma:Projectors2}
	Let $\ket{\phi}$ a quantum pure state, $P \le \I$ a projector acting on $\ket{\phi}$ and $\ket{\psi}$ such that $P \ket{\psi} = \ket{\psi}$. We have $|\braket{\phi}{\psi}|^2 \le tr(P \kb{\phi})$.
\end{lemma}
\begin{proof}
	We decompose $\ket{\phi}$ in order to make $\ket{\psi}$ appear. We write $\ket{\phi} = \alpha \ket{\psi} + \beta \ket{\psi^\bot}$ with $|\alpha|^2 + |\beta|^2 = 1$ and $\braket{\psi}{\psi^\bot} = 0$. This gives us $P \ket{\phi} = \alpha \ket{\psi} + \beta P \ket{\psi^\bot}$. Notice that we also have $\triple{\psi}{P}{\psi^\bot} = 0$. From there, we conclude
	$$ tr(P \kb{\phi}) = ||P \ket{\phi}||^2 = |\alpha|^2 + |\beta|^2 ||P \ket{\psi^\bot}||^2 \ge |\alpha|^2 = |\braket{\phi}{\psi}|^2.$$
\end{proof}
\subsection{Single outcome case : S = 1}
We first prove the theorem for the case where $S=1$. 

\begin{theorem}\label{Theorem:Single}
	Consider $n$ projectors $P_1,\dots,P_n$ and a quantum mixed state ${\sigma}$ in some Hilbert space $\spa{B}$. Let $V := \frac{1}{n} \sum_{i=1}^n tr(P_i \sigma)$, and let 
	$$E := \frac{1}{n(n-1)} \sum_{i,j\neq i} tr(P_j P_i  \sigma P_i P_j).$$ 
	Then it holds that $E \ge \frac{1}{64}\big(V - \frac{1}{n}\big)^3$.
\end{theorem}
\begin{proof}

We fix a quantum mixed state $\sigma$ in some Hilbert space $\spa{B}$ and $n$ projectors $P_1,\dots,P_n$ acting on $\spa{B}$. We first move to the realm of pure states which will be easier to analyze by adding an extra Hilbert space $\spa{E}$. We consider a purification $\ket{\phi}$ of $\sigma$ in some space $\spa{BE}=\spa{B}\otimes \spa{E}$. We define 
\begin{align*}
	\ket{\phi_i} := \frac{(P_i \otimes \mathbbm{1}_{\mathcal{E}}) \ket{\phi}}{||(P_i \otimes \mathbbm{1}_{\mathcal{E}}) \ket{\phi}||}.
\end{align*}
The state $\ket{\phi_i}$ corresponds to the normalized projection of $\ket{\phi}$ using $P_i$. We first express $E$ and $V$ as inner products of the quantum pure states we defined: 
\begin{lemma} 
	$ E \ge \frac{1}{n(n-1)} \sum_{i,j \neq i} |\braket{\phi}{\phi_i}|^2 |\braket{\phi_i}{\phi_j}|^2 \quad \mbox{and} \quad V = \frac{1}{n} \sum_{i=1}^n |\braket{\phi}{\phi_i}|^2.$
\end{lemma}
\begin{proof}
	We write 
	\begin{align*}
		E & = \frac{1}{n(n-1)} \sum_{i,j\neq i} tr(P_j P_i \sigma P_i P_j ) \\
		& = \frac{1}{n(n-1)} \sum_{i,j\neq i} tr\left( (P_j \otimes \mathbbm{1}_{\mathcal{E}})(P_i \otimes \mathbbm{1}_{\mathcal{E}})  \kb{\phi} (P_i \otimes \mathbbm{1}_{\mathcal{E}}) (P_j \otimes \mathbbm{1}_{\mathcal{E}}) \right) 
	\end{align*}
	Here, by using Lemma \ref{Lemma:Projectors1}, notice that 
	$$ (P_i \otimes \mathbbm{1}_{\mathcal{E}})\kb{\phi}(P_i \otimes \mathbbm{1}_{\mathcal{E}}) = ||(P_i \otimes \mathbbm{1}_{\mathcal{E}})\ket{\phi}||^2 \kb{\phi_i} = |\braket{\phi}{\phi_i}|^2 \kb{\phi_i}.$$
	From there, we can continue have
	\begin{align*}
		E & = \frac{1}{n(n-1)} \sum_{i,j\neq i}  |\braket{\phi}{\phi_i}|^2 tr\left( (P_j \otimes \mathbbm{1}_{\mathcal{E}}) \kb{\phi_i}  \right) \\
		& \ge \frac{1}{n(n-1)} \sum_{i,j \neq i} |\braket{\phi}{\phi_i}|^2 |\braket{\phi_i}{\phi_j}|^2 
	\end{align*}
	where the last inequality comes from Lemma \ref{Lemma:Projectors2}. 
	Notice also that we immediately have $V = \frac{1}{n} \sum_{i=1}^n tr(P_i \sigma) = \sum_{i} |\braket{\phi}{\phi_i}|^2$. 
\end{proof}
Our goal is to relate $E$ and $V$. We will deal with the terms $|\braket{\phi_i}{\phi_j}|^2$ using the following proposition on almost orthogonal states.
\begin{proposition}\label{Proposition:AlmostOrthogonal}
	Consider $n$ quantum pure states $\ket{\phi_1},\dots,\ket{\phi_n}$. Let 
	$$ S := \max_{\ket{\Omega}} \sum_{i=1}^n |\braket{\Omega}{\phi_i}|^2 \quad \text{and} \quad C := \sum_{i,j \neq i}^n |\braket{\phi_i}{\phi_j}|^2.$$ 
	We have $S \le 1 + \sqrt{\frac{(n-1)C}{n}} \le 1 + \sqrt{C}$.
\end{proposition}
\begin{proof}
	Let $M = \sum_{i=1}^n \kb{\phi_i}$. $M$ is a positive semi-definite matrix of dimension at most $n$. Let $\lambda_1 \ge \lambda_2 \ge \dots \ge \lambda_n$ the $n$ eigenvalues of $M$ in decreasing order. We have $\sum_i \lambda_i = tr(M) = \sum_j tr(\kb{\phi_j}) = n$. Moreover, notice that $S = \max_{\ket{\Omega}} \sum_i |\braket{\Omega}{\phi_i}|^2 = \lambda_1$.
	
	We write $M^2 = \sum_{i,j} \braket{\phi_i}{\phi_j} \ketbra{\phi_i}{\phi_j}$ and $tr(M^2) = \sum_{i,j} |\braket{\phi_i}{\phi_j}|^2 = n + C$. Moreover, we have $tr(M^2) = \sum_i \lambda_i^2$. This gives us
	\begin{align*}
		n + C & = tr(M^2) = \sum_{i=1}^{n} \lambda_i^2 = \lambda_1^2 + \sum_{i=2}^{n} \lambda_i^2 
		\ge \lambda_1^2 + (n-1) \left(\frac{n - \lambda_1}{n-1}\right)^2 \\
		&= \lambda_1^2 + \frac{(n-\lambda_1)^2}{n-1} = S^2 + \frac{(n-S)^2}{n-1}
	\end{align*}
	where the inequality comes from the convexity of the square function. 
	From there, we have 
	\begin{align*}
		(n-1)S^2 + (n-S)^2 - n(n-1) \le (n-1)C 
	\end{align*}
	Using $(n-1)S^2 + (n-S)^2 - n(n-1) = n(S-1)^2$, we conclude that $n(S-1)^2 \le (n-1)C$ or equivalently
	$S \le 1 + \sqrt{\frac{(n-1)C}{n}}$.
\end{proof}
In particular, the above proposition implies that 
$$V \le \frac{1}{n} + \frac{n-1}{n} \sqrt{\frac{1}{n(n-1)}\sum_{i,j\neq i} |\braket{\phi_i}{\phi_j}|^2}.$$
The term in the squared root is very similar to $E$. Unfortunately, the expression for $E$ contains an extra factor $|\braket{\phi}{\phi_i}|^2$ in the sum under the square-root. If the quantity $|\braket{\phi}{\phi_i}|^2$ was independent of $i$, it would be equal to $V$ and we would be able to conclude. However, this is not always the case and this adds a difficulty in the proof. In order to overcome it, we will use Proposition \ref{Proposition:AlmostOrthogonal} only with the states for which $|\braket{\phi}{\phi_i}|^2$ is not too small. We will choose a threshold $\kappa$ (that will be fixed later) and consider only the indices $i$ for which $|\braket{\phi}{\phi_i}|^2 \ge V/\kappa$. This is the goal of the next proposition.
\begin{proposition}\label{Proposition:MainProposition}
	$ \forall \kappa > 1, \ V \le \left(1 + \frac{1}{\kappa -1}\right) \left(\frac{1}{n} + \sqrt{\frac{\kappa E}{V}}\right). $
\end{proposition}
\begin{proof}
	For all $i$, let $p_i := |\braket{\phi}{\phi_i}|^2$. We have by definition $V = \sum_i p_i$. We fix $\kappa > 1$ and define the set $Z := \{i \in [n] : p_i \ge \frac{V}{\kappa}\}$. We have 
	$$ \frac{1}{n} \sum_{i \notin Z} p_i \le \frac{1}{n} \sum_{i \notin Z} \frac{V}{\kappa}  \le \frac{V}{\kappa}, $$
	which implies 
	\begin{align}\label{Eq10}
		\frac{1}{n} \sum_{i \in Z} p_i \ge (1 - \frac{1}{\kappa})V.
	\end{align}
	We write 
	\begin{align}\label{Eq2}
		E & \ge \frac{1}{n(n-1)} \sum_{i,j \neq i} p_i |\braket{\phi_i}{\phi_j}|^2 \ge\frac{1}{n(n-1)} \sum_{\substack{i,j \in Z\\ i \ne j}} p_i |\braket{\phi_{i}}{\phi_j}|^2 \\
		&\ge \frac{V}{\kappa} \cdot \frac{1}{n(n-1)} \sum_{\substack{i,j \in Z\\ i \ne j}}  |\braket{\phi_{i}}{\phi_j}|^2.
	\end{align}
	Now, starting from Equation \ref{Eq10}, we have 
	\begin{align}
		V & \leq \frac{1}{1 - \frac{1}{\kappa}} \frac{1}{n} \sum_{i \in Z} p_i = \frac{1}{1 - \frac{1}{\kappa}} \frac{1}{n} \sum_{i \in Z} |\braket{\phi}{\phi_i}|^2 \le \Big(1 + \frac{1}{\kappa - 1}\Big) \max_{\ket{\Omega}} \frac{1}{n} \sum_{i \in Z} |\braket{\Omega}{\phi_i}|^2 \nonumber \\ 
		&\le \Big(1 + \frac{1}{\kappa - 1}\Big)  \left(\frac{1}{n} + \sqrt{\frac{1}{n(n-1)} \sum_{\substack{i,j \in Z\\ i \ne j}}  |\braket{\phi_{i}}{\phi_j}|^2 }\right)  \label{eq1} \\
		&\le \Big(1 + \frac{1}{\kappa - 1}\Big)  \left(\frac{1}{n} + \sqrt\frac{\kappa E}{V}\right)
	\end{align} 
where we used Lemma \ref{Proposition:AlmostOrthogonal} in Equation \ref{eq1} and Equation \ref{Eq2} for the last inequality. This proves the proposition.
\end{proof}
We can now use Proposition \ref{Proposition:MainProposition} to prove our theorem. We distinguish two cases: 
\begin{enumerate}
	\item If $(\frac{V}{n^2E})^{1/3} > 2$. We take $\kappa = (\frac{V}{n^2E})^{1/3} > 2$ which implies $\kappa (\frac{n^2 E}{V})^{1/3} = 1$ and $(\kappa \frac{n^2 E}{V})^{\frac{1}{2}} = \frac{1}{\kappa} $. We get 
	\begin{align*}
		V & \le \frac{1}{n} \Big(1 + \frac{1}{\kappa - 1}\Big) \left( 1 +  \Big[{\kappa\frac{n^2 E}{V}} \Big]^{1/2} \right) = \frac{1}{n} \Big(1 + \frac{1}{\kappa - 1}\Big)\Big(1 + \frac{1}{\kappa }\Big) \\ & \le \frac{1}{n} (1 + \frac{4}{\kappa}) = \frac{1}{n} + 4\left(\frac{E}{nV}\right)^{1/3}.
	\end{align*} 
	This gives $E \ge \frac{nV}{64}(V - \frac{1}{n})^3$ which implies $E \ge \frac{1}{64}(V - \frac{1}{n})^3$. To see this last implication, consider the following two cases: if $V \ge \frac{1}{n}$ then the equality comes immediately from the previous inequality. If $V \le \frac{1}{n}$, we immediately have $E \ge 0 \ge \frac{1}{64}(V - \frac{1}{n})^3$.
	\item If $(\frac{V}{n^2E})^{1/3} \le 2$. This implies $\left(\frac{V}{E}\right)^{1/2} \le n \cdot 2^{3/2}$. We take $\kappa = 2$ and obtain 
	\begin{align*}
		V &\le \Big(1 + \frac{1}{\kappa - 1}\Big)  \left(\frac{1}{n} + \sqrt\frac{\kappa E}{V}\right) 
		= 2 (\frac{1}{n} + \sqrt{\frac{2E}{V}})\\
		& \le 2(2^{2/3}\sqrt{\frac{E}{V}} + \sqrt{\frac{2E}{V}}) \le 6\sqrt{\frac{E}{V}}
	\end{align*}
	which implies $E \ge \frac{V^3}{36} \ge \frac{1}{64}(V - \frac{1}{n})^3$.
\end{enumerate}
\end{proof}
\subsection{General case}
We can now show our theorem for any $S$.
The general case will be a direct corollary of the following.
\begin{proposition}\label{Proposition:MultipleProj}
	Let a projector $P := \sum_{i = 1}^m P_i$ where $\{P_i\}_{i \in [m]}$ are orthogonal projectors. For any pure state $\ket{\psi}$, we have
	$$\sum_{i = 1}^m P_i   \kb{\psi}   P_i \ge \frac{1}{m} P   \kb{\psi} P.
	$$
\end{proposition}
We note that this result can be obtained as an application of the pinching inequality \cite{Hay02,SBT16}, but we provide a proof here for completeness.
\begin{proof}
	We define the following \emph{unnormalized states} $\ket{\psi^P} = P(\ket{\psi})$ and $\ket{\psi^P_i} = P_i(\ket{\psi})$. Because $P = \sum_i P_i$, we have $\ket{\psi^P} = \sum_i \ket{\psi^P_i}$. This gives 
	\begin{align*}
		\sum_{i = 1}^m P_i  \kb{\psi}  P_i  = \sum_{i = 1}^m \kb{\psi^P_i} \\
		P  \kb{\psi}  P^{\dagger} = \kb{\psi^P}
	\end{align*}Consider now any state $\ket{\phi} = \sum_i \alpha_i \ket{\psi^P_i} + \ket{\xi}$ where $\ket{\xi}$ 
	is orthogonal to all the $\ket{\psi^P_i}$. We have
	\begin{align*}
		\triple{\phi}{\sum_{i = 1}^m P_i  \kb{\psi}  P_i}{\phi} = \sum_i |\braket{\psi^P_i}{\phi}|^2 = |\alpha_i|^2\left|\braket{\psi^P_i}{\psi^P_i}\right|^2
	\end{align*}
	and 
	\begin{align*}
		\triple{\phi}{  P \kb{\psi} P  }{\phi} = |\braket{\psi^P}{\phi}|^2 = \left|\sum_i \alpha_i \braket{\psi^P_i}{\psi^P_i}\right|^2
	\end{align*}
	From there, we can conclude. We have:
	\begin{align*}
		\triple{\phi}{\sum_{i = 1}^m P_i \kb{\psi}  P_i}{\phi} & = \sum_i |\alpha_i|^2|\braket{\psi^P_i}{\psi^P_i}|^2 \\
		& \ge \frac{1}{m} \left|\sum_i |\alpha_i||\braket{\psi^P_i}{\psi^P_i}|\right|^2 & \text{(from  Cauchy-Schwarz)} \\
		& \ge \frac{1}{m} \triple{\phi}{ P \kb{\psi} P }{\phi} 
	\end{align*}
	Since this holds for any state $\ket{\phi}$, we can conclude that 
	$$\sum_{i = 1}^m P_i  \kb{\psi}   P_i \ge \frac{1}{m} P  \kb{\psi}  P^{\dagger}.
	$$
\end{proof}
From there, and using the previous theorem, we can show our main technical result.
\begin{reptheorem}{Theorem:Multi}
	Consider $n$ projectors $P_1,\dots,P_n$ such that for each $i$, we can write $P_i := \sum_{s = 1}^S P_i^s$ where the $\{P_i^s\}_s$ are orthogonal projectors for each $i$, \textit{i.e.} for each $i$ and $s,s'$, we have $P^s_i P^{s'}_i = \delta_{s,s'} P^s_i$. Let $\sigma$ be any quantum state, let $V := \frac{1}{n} \sum_{i=1}^n tr(P_i  \sigma)$, and let \\ $E := \frac{1}{n(n-1)} \sum_{i,j\neq i} \sum_{s,s' = 1}^S tr(P^{s'}_j P^s_i \sigma  P^s_i P^{s'}_j)$. Then it holds that $E \ge \frac{1}{64S}\big(V - \frac{1}{n}\big)^3$.
\end{reptheorem}
\begin{proof}
	We fix $n$ projectors $P_1,\dots,P_n$ such that for each $i$, we can write $P_i := \sum_{s = 1}^S P_i^s$ where the $\{P_i^s\}_s$ are orthogonal projectors for each $i$. We fix a quantum state $\sigma$. We have 
	\begin{align*}
		E & = \frac{1}{n(n-1)} \sum_{i,j\neq i} \sum_{s,s' = 1}^S tr(P^{s'}_j P^s_i \sigma (P^{s'}_i) (P^s_j) ) \\
		& = \frac{1}{n(n-1)} \sum_{i,j\neq i} \sum_{s=1}^S tr(P_j P^s_i \sigma (P^{s'}_i) P_j) \\
		& \ge \frac{1}{Sn(n-1)} \sum_{i,j\neq i}  tr(P_j P_i \sigma (P_i) P_j) & \textrm{(from Proposition } \ref{Proposition:MultipleProj}) \\
		& \ge \frac{1}{64S} \Big(V - \frac{1}{n}\Big)^3 & \mbox{(from Theorem } \ref{Theorem:Single} )
	\end{align*}
\end{proof}

\section{Entangled games}\label{Section:EntangledGames}
The goal of this section is to use the consecutive measurement theorems of the previous section to establish upper bounds on the value of entangled  games. For a game $G$ on the uniform distribution, we will define a game $G_{coup}$ which corresponds to a couple of instances of $G$ where Alice plays twice with the same input and Bob receives two distinct inputs and they need to win both instances in order to win the game  $G_{coup}$. In the cases we consider, upper bounding $G_{coup}$ will be easily done from non-signaling. Our learning lemmata will allow us to relate the winning probabilities of $G$ and $G_{coup}$. These two steps together will give us bounds on the value of $G$.
\subsection{First definitions}
\begin{definition}
A game $G = (I_A,I_B,O_A,O_B,V,p)$ is defined by 
\begin{itemize}
\item 2 input sets $I_A,I_B$ which are respectively Alice's and Bob's input sets.
\item 2 output sets sets $O_A,O_B$ which are respectively Alice's and Bob's output sets.
\item A valuation function $V : I_A \times I_B \times O_A \times O_B \rightarrow \zo$ which indicates whether the game is won for some fixed input and outputs. The game is won if the value of $V$ is $1$.
\item A probability function $p : I_A \times I_B \rightarrow [0,1]$ which corresponds to the input distribution. We have $\sum_{(x,y) \in I_A \times I_B} p_{xy} = 1$.
\end{itemize}
\end{definition}
\begin{definition}
	A game $G = (I_A,I_B,O_A,O_B,V,p)$ is said to be \emph{on the uniform distribution} if $\forall (x,y) \in I_A \times I_B, \ p_{xy} = \frac{1}{|I_A||I_B|}$.
\end{definition}
\begin{definition}
	A game $G = (I_A,I_B,O_A,O_B,V,p)$ is \emph{projective} if
	$$\forall (x,y) \in  I_A \times I_B \ st. \ p_{xy} \neq 0, \ \forall a \in O_A, \ \exists! \ b \in O_B, \ st. \ V(x,y,a,b) = 1.$$
	A game $G$ is $S$-\emph{projective} if
	$$\forall (x,y) \in  I_A \times I_B \ st. \ p_{xy} \neq 0, \forall a \in O_A, |\{b \in O_B : V(x,y,a,b) = 1\}| \le S.$$
	In particular, a projective game is $1$-projective.
\end{definition}
In the case where Alice and Bob are classical and want to win a game $G$, it is known that their optimal strategy to win is to perform a deterministic strategy. Notice that a projective game is asymmetric in Alice and Bob.
\begin{definition}
	For a game $G = (I_A,I_B,O_A,O_B,V,p)$, we denote by $\omega^*(G)$ its entangled value, \textit{i.e.} the maximum winning probability for the game when Alice and Bob are quantum and share an entangled state. 
\end{definition}
In order to study this maximal winning probability, it is enough to consider the case where Alice and Bob perform projective measurements.

In order to prove upper bounds on $\omega^*(G)$ for a game $G$ on the uniform distribution, we introduce the notion of coupled game $G_{coup}$. 
\begin{definition}
	For any game $G = (I_A,I_B,O_A,O_B,V,p)$ on the uniform distribution we define $G_{coup}$ as follows:
	\begin{itemize}
		\item Alice receives a random $x \in_R I_A$. Bob receives a random pair of \emph{different} inputs $(y,y')$ from $I_B$. 
		\item Alice outputs $a \in O_A$. Bob outputs $b,b' \in O_B$.
		\item They win the game if $V(x,y,a,b) = V(x,y',a,b') = 1$.
	\end{itemize}
\end{definition}

\subsection{Relating $G$ and $G_{coup}$}
In this section, we use our results from the previous section to relate the values of $G$ and $G_{coup}$. 
\begin{repproposition}{Proposition:Quantum_coup_bound}
	For any game $G$ on the uniform distribution which is $S$-projective, we have $\omega^*(G_{coup}) \ge \frac{1}{S \cdot 64} (\omega^*(G) - \frac{1}{n})^3$ where $n = |I_B|$.
\end{repproposition}
\begin{proof}
	Consider an optimal strategy for Alice and Bob for the game $G$. In particular, for each $y$, let $Q^y = \{Q^y_b\}$ the projective measurement that corresponds to his strategy for input $y$. Fix an input/output pair $(x,a)$ for Alice and let $\sigma^{xa}$ be the state held by Bob, conditioned on this pair. For each $y$, let $W_y = \{b : V(a,b|x,y) = 1\}$ be the set of winning outputs for Bob. Since $G$ is $S$-projective, we have $|W_y| \le S$. We define $Q^y_W = \sum_{b \in W_y} Q^y_b$.
	
	We denote by $V^{xa}$ the probability that Alice and Bob win the game for a fixed $x,a$. Notice that $\omega^*(G) = \E_{xa}[V^{xa}]$. We have 
	$$ V^{xa} =  \frac{1}{n}\sum_y tr(Q^y_W \sigma^{xa} (Q^y_W)), $$
since $y$ is uniformly distributed over the set $I_B$ of size $n$.
	
	We now consider the following quantum strategy for $G_{coup}$: Alice and Bob share the same initial state as in the optimal strategy for $G$; Alice performs the same measurement strategy as for $G$; on inputs $y,y$, Bob applies the first measurement $Q^y$ and obtains outcome $b$, then applies the measurement $P^{y'}$ on his resulting state and gets outcome $b'$. Bob outputs $(b,b')$. Let $E^{xa}$ be the probability that Alice and Bob win $G_{coup}$ using this strategy for a fixed $x,a$. Notice that $\omega^*(G_{coup}) \ge \E_{xa}[E^{xa}]$ since the value $\E_{xa}[E^{xa}]$ is achievable. We have
	\begin{align*}
		E^{xa} & = \frac{1}{n(n-1)} \sum_{y,y' \neq y} 
		\sum_{\substack{b : V(ab|xy) =1  \\ b' : V(ab'|xy') = 1}}
		tr(Q^{y'}_{b'}Q^y_b \sigma^{xa} Q^y_b Q^{y'}_{b'}) \\
		& \ge Pos(\frac{1}{64S}(V^{xa} - \frac{1}{n})^3) & \mbox{from Theorem } \ref{Theorem:Multi}
	\end{align*}
where $Pos(x) := \max(x,0)$ is the positive part of $x$.
	By taking the expectation on each side, we obtain	 
	\begin{align*}
	 \omega^*(G_{coup}) &= \E_{xa}[E^{xa}] \ge \E_{xa}[Pos(\frac{1}{64S}(V^{xa} - \frac{1}{n})^3)] \ge Pos(\frac{1}{64S}(\omega^*(G) - \frac{1}{n})^3) \\
	&\ge \frac{1}{64S}(\omega^*(G) - \frac{1}{n})^3
	\end{align*}
	where we used the convexity of the function $x \mapsto Pos(x^3)$.
\end{proof}
\subsection{Retrieving the value of certain entangled games}\label{subSection:EntangledGames}
We now use the technique developed above in order to obtain upper bounds on games based on the $\FQ$ variant of ${CHSH}$.

$\mathbf{CHSH^Q(P)}$ --- 
We consider the nonlocal game called ${CHSH^Q(P)}$ with $P \leq Q$. Here, Alice and Bob receive inputs $x$ and $y$, where $x$ is a uniformly random element in $\F_Q$ and $y$ is an element of $\F_Q$ taken uniformly at random from $\{0,\dots,P-1\}$. They output values $a, b\in \FQ$ and win if $a + b = x * y$, where the addition and multiplication are with respect to $\FQ$. Notice that ${CHSH^Q(P)}$ is a projective game on the uniform distribution. 

Let's analyze ${CHSH^Q(P)_{coup}}$. Fix an input/output pair $(x,a)$ and a pair $(y,y')$ of inputs for Bob with $y \neq y'$.
Let $b,b'$ Bob's output. If Alice and Bob win the game then we have $a + b = x * y$ and $a + b' = x * y'$ which implies that $(b-b')*(y - y')^{-1} = x$. This means that Bob can use any strategy for ${CHSH^Q(P)_{coup}}$ as a strategy to guess $x$ with the same winning probability. Because of non-signaling, this happens with probability at most $\frac{1}{Q}$. We therefore have
$\omega^*(CHSH^Q(P)_{coup}) \le \frac{1}{Q}$. Using Proposition \ref{Proposition:Quantum_coup_bound} (we have $S = 1$ in this setting), we obtain 
$\omega^*(CHSH^Q(P)) \le \frac{1}{P} + \frac{4}{Q^{1/3}}$. \\ \\
$\mathbf{CHSH^Q(2)}^{\otimes n}$ --- This is the parallel repetition of $\CHSH^Q$ where Alice and Bob receive $n$ uniform strings $x_1, \cdots, x_n$ and $y_1, \cdots, y_n \in \zo$ and output strings $a_1, \cdots, a_n$ and $b_1, \cdots, b_n$, respectively. They win the $\CHSH^Q(2)^{\otimes n}$  game if they win all $n$ instances of the $\CHSH^Q$ games, \textit{i.e.} if $a_i + b_i = x_i * y_i$ for all $i \in \{1, \cdots, n\}$. Consider now the coupled version of this game. For any two inputs $y=y_1,\dots,y_n,y'=y'_1,\dots,y'_n$ given to Bob, if Alice and Bob win the game then similarly as in $\CHSH$, Bob can recover Alice's input bits $x_i$ for each $i$ where $y_i \neq y'_i$. From non signaling, this happens with probability at most $Q^{-|y-y'|_H}$, where $|y-y'|_H$ is the Hamming distance between strings $y$ and $y'$, counting in how many indices both strings differ. Therefore, we have 
$\omega^*(CHSH^Q(2)_{coup}^{\otimes n}) = \E_{y,y'\neq y}[Q^{-|y-y'|_H}] = \frac{1}{2^n}\left((1 + \frac{1}{Q})^n - 1\right)$. If $Q > n$, we have 
$$ \omega^*(CHSH^Q(2)_{coup}^{\otimes n}) \le \frac{2n}{Q 2^n} \quad \mbox{which gives} \quad
\omega(CHSH^Q(2)^{\otimes n}) \le \frac{1}{2^n} + 4(\frac{2n}{Q 2^n})^{1/3}.$$
In particular, if we take $Q = \frac{64 \cdot 2^{2n}}{2n\eps^3}$, we obtain $\omega(CHSH^Q(2)^{\otimes n}) \le \frac{1}{2^n} \left(1 + \eps\right)$.

\section{Relativistic bit and string commitment} \label{Section:RelativisticBC}
In this section, we will review the relativistic $\F_Q$ bit commitment scheme and its natural extension to string commitment. We will show how the sum-binding property (with worst parameters) is preserved when considering string commitment or the parallel repetition of bit commitment. This is showed by Propositions \ref{Proposition:P-Sum-binding} and \ref{Proposition:Parallel_Sum-binding}.

\subsection{Bit commitment}
\emph{Bit commitment} is a cryptographic primitive between two distrustful parties Alice and Bob which consists of $2$ phases: a \emph{Commit phase} and a \emph{Reveal phase}. Alice has a bit $d$ at the beginning of the protocol. In the commit phase, Alice will commit to this value $d$ by performing some communication protocol such that at end of the commit phase, Bob has no information about $d$ (hiding property). In the second phase, the reveal phase, Alice and Bob also perform some communication which results in Alice revealing $d$.  A desired property here is that Alice is unable to reveal a bit different from the one chosen during the commit phase (binding property). 

In some sense, a bit commitment protocol simulates a digital safe. In the commit phase, Alice writes her input $d$ on a piece of paper, puts that paper into the safe and sends the safe to Bob. If Bob doesn't hold the key of the safe then he cannot open it and therefore has no information about $d$. In the reveal phase, Alice would send to Bob the key to open the safe. But she cannot change the value of the bit in the safe because Bob has control of the safe. This primitive has been widely studied. However,  bit commitment can only be performed with computational security in the most usual models. 

We now define more formally a bit commitment scheme.

\begin{definition}
	A quantum commitment scheme  is an interactive protocol between Alice and Bob with two phases, a Commit phase and a Reveal phase.
	
	\begin{itemize}
		\item \emph{Commit phase}. Alice chooses a uniformly random input $d$ that she wants to commit to. To do so, Alice and Bob perform a communication protocol that corresponds to this commit phase. 
		\item \emph{Reveal phase}. Alice interacts with Bob in order to reveal $d$. To do so, they perform a second communication protocol where at the end, Bob should know the value revealed by Alice. Bob, depending on this revealed value and the interaction with Alice, outputs either ``Accept" or ``Reject".  
	\end{itemize}
\end{definition}

A commitment scheme $\Pi = (COMM,OPEN)$ is the description of the protocol followed by the honest parties during both the commit and the open phases. All protocols 
that we will consider will be perfectly hiding and we will only be interested in the binding property. Therefore, we only consider the case of a cheating Alice, 
which will be described through her cheating strategy  $\Str^* = (\Comm^*,\Open^*)$ in both phases of the protocol. The binding property we consider is the standard sum-property, that was also used 
in previous work regarding relativistic bit commitment \cite{LKB+15,FF15,CCL15}. 

\begin{definition}[Sum-binding]
	\label{def1}
	We say that a bit commitment protocol $\Pi$ is \emph{$\eps$-sum-binding} if
	$$ \forall \ \Comm^*, \ \sum_{d = 0}^{1} \max_{\Open^*} \left( \Pr[\mbox{Alice successfully reveals } d \mid (\Comm^*,\Open^*)]  \right) \le 1 + \eps.$$
\end{definition}

In the case of string commitment, meaning Alice wants to commit/reveal to a string of dimension $P$ (\ie $\lceil\log(P)\rceil bits$), we can extend the sum-binding property as follows.
\begin{definition}[String sum-binding]
	\label{def21}
	We say that a $P$-string commitment protocol $\Pi$ is \emph{$\eps$-sum-binding} if
	$$ \forall \ \Comm^*, \ \sum_{d = 0}^{P-1} \max_{\Open^*} \left( \Pr[\mbox{Alice successfully reveals } d \mid (\Comm^*,\Open^*)]  \right) \le 1 + \eps.$$
\end{definition}

The sum-binding property for bit commitment is a relatively weak one. Indeed, it is very hard to use this definition when combining it with other primitives. For example, when committing to $n$ bits in parallel, it is not always the case that this overall commitment, seen as a $2^n$-string commitment, satisfies a good string sum-binding property. On the other hand, the string sum-binding for strings seems more exploitable.

\subsection{Relativistic bit commitment}\label{Section:DescriptionOfTheProtocol}
A relativistic bit commitment scheme is a commitment scheme where we use physical property that no information carrier can travel faster than the speed of light. In order to take advantage of this principle, we split Alice (resp.~Bob) into $2$ agents $\mathcal{A}_1$ and $\mathcal{A}_2$ (respectively $\mathcal{B}_1$ and $\mathcal{B}_2$). For each $i \in \{1,2\}$, $\mathcal{A}_i$ interacts only with $\mathcal{B}_i$. If we put the two pairs $(\mathcal{A}_1, \mathcal{B}_1)$ and $(\mathcal{A}_2,\mathcal{B}_2)$ far apart, and use some timing constraints, we can enforce some non-signaling type scenarios. Here, we will only use the property that the two honest Bob's know their respective location. In particular, there is no trust needed regarding the location of the cheating parties. 

The security definitions for relativistic bit commitment are the ones we presented above: definitions \ref{def1} and \ref{def21}. We will now describe the $\F_Q$ relativistic bit commitment scheme. This scheme will consist of $4$ phases, the preparation phase, the commit phase, the sustain phase and the reveal phase. The preparation phase is some preprocessing phase that can be done anytime before the protocol. The sustain phase can be seen as a part of the reveal phase, and corresponds to the time where the committed bit is safe. We assume here that the two Alices learn at the beginning of the sustain phase the bit $d$ they should try to reveal (which doesn't necessarily correspond to the bit, if any, they committed to).

{\bf The single-round $\F_Q$ protocol}. --- The single-round version corresponds $\CHSHQ$ to the protocol introduced by Cr{\'e}peau \etal \cite{CSST11} (see also \cite{sim07}).
Both players, Alice and Bob, have agents $\mathcal{A}_1, \mathcal{A}_2$ and $\mathcal{B}_1, \mathcal{B}_2$ present at two spatial locations, 1 and 2, separated by a distance $D$. We consider the case where Alice makes the commitment. 
The protocol (followed by honest players) consists of 4 phases: preparation, commit, sustain and reveal. The sustain phase in the single-round protocol is trivial and simply consists in waiting for a time less than $D/c$, which is the time needed for light to travel between the two locations. The bit commitment protocol goes as follows. 
\begin{enumerate}
	\item \emph{Preparation phase}: $\mathcal{A}_1,\mathcal{A}_2$ (resp.~$\mathcal{B}_1,\mathcal{B}_2$) share a random number $a\in \FQ$ (resp.~$x \in \FQ$).
	\item \emph{Commit phase}: $\mathcal{B}_1$ sends $b$ to $\mathcal{A}_1$, who immediately returns $y = a + d *x$ where $d \in \zo$ is the committed bit. 
	\item \emph{Sustain phase}: $\mathcal{A}_1$ and $\mathcal{A}_2$ wait for some time $\tau < D/c$, where $c$ is the speed of light. Crucially, for any time less than $D/c$, the NSS principle guarantees that $\mathcal{A}_2$ has no information about the value of $b$.
	\item \emph{Reveal phase}: $\mathcal{A}_2$ reveals the values of $d$ and $a$ to $\mathcal{B}_2$ who checks that $y = a + d *t x$.
\end{enumerate}

This relativistic bit commitment protocol is known to be $O(\frac{1}{\sqrt{Q}})$-sum-binding \cite{LKB+15}.
It can be easily extended to a $P$-string commitment where $d$ is an element of $\F_P$ instead of an element of $\zo$. The above construction is well defined as long as $Q \ge P$ (all the operations are still the modular operations in $\F_Q$).

\begin{proposition} \label{Proposition:P-Sum-binding}
	The above relativistic $P$-string commitment protocol is $\eps$-sum-binding with $\eps = \frac{4P}{Q^{1/3}}$.
\end{proposition}
\begin{proof}
	Consider a $P$-string commitment $\Pi$ and a cheating strategy $Str^* = (Comm^*,Open^*)$. In this strategy, $\mathcal{A}_1$ and $\mathcal{A}_2$ share an entangled state $\ket{\psi}$. After receiving $b$, $\mathcal{A}_1$ performs a measurement on her part of the state to produce an output $y$ which she sends to $\mathcal{B}_1$. For a random $d$ that $\mathcal{A}_2$ wants to reveal, she performs a measurement on her part of the state to produce an output $a$. We have 
	$$ \frac{1}{P} \sum_{d = 0}^{P-1} \left( \Pr[\mbox{Alice successfully reveals } d \mid (\Comm^*,\Open^*)]  \right) = \Pr[a + y = b * d].$$
	
	One can directly use the above strategy to construct a strategy for a $\CHSHQ(P)$ game (defined in Section \ref{Section:EntangledGames}), with respective inputs $b \in F_Q$, $d \in F_P$ and with respective outputs $y$ and $a$. We have immediately
	$$ \Pr[a + y = b * d] \le \omega^*(\CHSHQ(P)) \le \frac{1}{P} + \frac{4}{Q^{1/3}},$$
	where the bound on the entangled is the one from Section \ref{Section:EntangledGames}. This gives us 
	$$ \sum_{d = 0}^{P-1} \max_{\Open^*} \left( \Pr[\mbox{Alice successfully reveals } d \mid (\Comm^*,\Open^*)]  \right) \leq 1 + \frac{4P}{Q^{1/3}}$$
which proves the desired proposition.
\end{proof}
If we want to perform an $\eps$-sum-binding $P$-string commitment protocol then we need to send $\log(Q) = \log(\frac{64P^3}{\eps^3}) = 3(\log(P) + |\log(\eps)|) + 8$ bits for each round of the protocol. 

\subsection{Parallel repetiton of $RBC$}
The problem with string commitment is that it is not possible to reveal only some bits of the string: by construction, one has to reveal the whole string. In order to circumvent this issue, we need to consider performing a bit commitment $n$ times in parallel. This then allows one to reveal only a fraction of the bits. The scheme will still feature sum-binding property but the scaling in parameters -- although still polynomial -- will not be as good as for string commitment.

\begin{enumerate}
	\item \emph{Preparation phase}: $\mathcal{A}_1,\mathcal{A}_2$ (resp.~$\mathcal{B}_1,\mathcal{B}_2$) share $n$ random bits $a_1,\dots,a_n \in \FQ$ (resp.~$b_1,\dots,b_n \in \FQ$).
	\item \emph{Commit phase}: $\mathcal{B}_1$ sends each $b_i$ to $\mathcal{A}_1$, who returns for each $i$ $y_i = a_i + d_i * b_i$ where $d_1,\dots,d_n \in \zo$ is the sequence of committed bits. 
	\item \emph{Sustain phase}: $\mathcal{A}_1$ and $\mathcal{A}_2$ wait for some time $\tau \leq D/c$.
	\item \emph{Reveal phase}: Let $S$ be the subset of indices Alice wants to reveal. $\mathcal{A}_2$ indicates $S$ to $\mathcal{B}_2$ and reveals the values $\{a_i\}_{i \in S}$ and $\{d_i\}_{i \in S}$ to $\mathcal{B}_2$ who checks that for each $i \in S$, the relation $y_i = a_i + d_i*b_i$ holds.
\end{enumerate}

\begin{proposition} \label{Proposition:Parallel_Sum-binding}
	Fix a subset $S$ of indices Alice will reveal to. Relative to $S$, the above protocol  is $\eps$-sum binding with $\eps = 4(\frac{2|S|2^{2|S|}}{Q})^{1/3} \le 4(\frac{2n 2^{2n}}{Q})^{1/3}$. 
\end{proposition}
\begin{proof}
	Fix a subset $S$. As before, we can use a strategy for the relativistic bit commitment to solve an instance of $\CHSHQ(2)^{\otimes |S|}$ which implies 
	\begin{align*} \sum_{d \in \zo^{|S|}} \max_{\Open^*} \left( \Pr[\mbox{Alice successfully reveals } d \mid (\Comm^*,\Open^*)]  \right) \\
	\le 2^{|S|} \omega^*(\CHSHQ(2)^{\otimes |S|}).
	\end{align*}
	
	Since we know that $\omega^*(\CHSHQ(2)^{\otimes |S|}) \le \frac{1}{2^{|S|}} + 4(\frac{2|S|}{Q2^{|S|}})^{1/3}$, we can immediately conclude that 
	
	\begin{align*}
	 \sum_{d \in \zo^{|S|}} \max_{\Open^*} \left( \Pr[\mbox{Alice successfully reveals } d \mid (\Comm^*,\Open^*)]  \right) \le 2^{|S|}  \\
	 \le 1 + 4(\frac{2|S|2^{2|S|}}{Q})^{1/3}.\end{align*}
\end{proof}
If we want the above protocol to be $\eps$-sum-binding, we need to send $n \log(Q) = O(n^2 \log(n) + n|\log(\eps)|)$ bits at each round.

\section{Relativistic zero-knowledge}\label{Section:RelativisticZK}
In this section, we present our relativistic zero-knowledge protocol for $\mathsf{NP}$. Our protocol will be based on the well known protocol for the $\mathsf{NP}$-complete problem \HamiltonianCycle, which uses bit commitment.
\subsection{The zero-knowledge Hamiltonian cycle protocol}
Here, we present the zero-knowledge Hamiltonian cycle protocol and its adaptation to the relativistic setting. Let $S_n$ the set of permutation on $\{1,\dots,n\}$.
\begin{definition}
	A cycle of $\{1,\dots,n\}$ is a set of couples 
	$$\{(\Pi(1),\Pi(2)),(\Pi(2),\Pi(3)),\dots,(\Pi(n-1),\Pi(n)),(\Pi(n),\Pi(1))\}$$
	 for a permutation $\Pi \in S_n$. We denote by $\Gamma_n$ the set of cycles of $\{1,\dots,n\}$. We have $|\Gamma_n| = (n-1)!$. For a cycle $\mathcal{C} = \{(u,v)\}$ and a permutation $\Pi$, we also define $\Pi(\mathcal{C}) := \{(\Pi(u),\Pi(v)\}$
\end{definition}
\begin{definition}
	A Hamiltonian cycle of a graph $G = (V,E)$ is a cycle $\mathcal{C}$ of $\{1,\dots,|V|\}$ such that $\mathcal{C} \in E$ \ie $ \ \forall (i,j) \in \mathcal{C}, \ (i,j) \in E$.
\end{definition}
	Determining whether a graph $G$ has a Hamiltonian cycle or not is an NP-complete problem. The corresponding decision problem is \HamiltonianCycle and $G \in \HamiltonianCycle$ means that the graph contains a Hamiltonian cycle.
\subsection{The protocol}
We recall the zero-knowledge protocol for \HamiltonianCycle\ first presented by Blum \cite{Blu86}.

\fbox{
\begin{minipage}{0.93\textwidth}
\begin{center}	Zero knowledge protocol for \HamiltonianCycle  \end{center} 

\noindent \textbf{Input} --- The prover and the verifier are given a graph $G = (V,E)$. \\
\noindent \textbf{Auxiliary Input} --- The prover knows a Hamiltonian cycle $\mathcal{C}$ of $G.$\\
\noindent \textbf{Protocol} --- 
\begin{enumerate}
	\item The prover picks a random permutation $\Pi : V \rightarrow V$. He commits to each of bit of the adjacency matrix $M_{\Pi(G)}$ of $\Pi(G)$.
	\item The verifier sends a random bit (called the challenge) $chall \in \zo$ to the prover.
	\item
	\begin{itemize}
	\item If $chall = 0$, the prover decommits to all the elements of $M_{\Pi(G)}$, and reveals $\Pi$. 
	\item If $chall =1$, he reveals only the bits (of value $1$) of the adjacency matrix that correspond to a Hamiltonian cycle $\mathcal{C}' = \Pi(\mathcal{C})$ of $\Pi(G)$.
	\end{itemize}
	\item The verifier checks that these decommitments are valid and correspond, for $chall = 0$ to $M_{\Pi(G)}$ and, for $chall=1$, to a Hamiltonian cycle.
\end{enumerate}
\end{minipage} 
}

%

We now present the relativistic zero-knowledge protocol, that uses the $\F_Q$ bit commitment.

\fbox{
\begin{minipage}{0.93\textwidth}
\begin{center}	Relativistic zero knowledge protocol for \HamiltonianCycle\  \end{center} 

\noindent \textbf{Input} --- The provers and the verifiers are given a graph $G = (V,E)$. \\
\noindent \textbf{Auxiliary Input} --- The provers $P_1$ and $P_2$ know a Hamiltonian cycle $\mathcal{C}$ of $G$. \\
\noindent \textbf{Preprocessing} --- $P_1$ and $P_2$ agree beforehand on a random permutation $\Pi : V \rightarrow V$ and on an $n \times n$ matrix $A \in \mathcal{M}_n^{\F_Q}$ where each element of $A$ is chosen uniformly at random in $\F_Q$. \\
\noindent \textbf{Protocol} --- 
\begin{enumerate}
	\item Commitment to each bit of $M_{\Pi(G)}$ : $V_1$ sends a matrix $B \in \mathcal{M}_n^{\F_Q}$ where each element of $B$ is chosen uniformly at random in $\F_Q$. $P_1$ outputs the matrix $Y \in \mathcal{M}_n^{\F_Q}$ such that $\forall i,j \in [n], \ Y_{i,j} = A_{i,j} + (B_{i,j} * (M_{\Pi(G)})_{i,j})$.
	\item The verifier $V_2$ sends a random bit (called the challenge) $chall \in \zo$ to the prover $P_2$.
	\item
	\begin{itemize}
		\item If $chall = 0$, $P_2$ decommits to all the elements of $M_{\Pi(G)}$, \textit{i.e.} he sends all the elements of $A$ to $V_2$  and reveals $\Pi$. 
		\item If $chall =1$, $P_2$ reveals only the bits (of value $1$) of the adjacency matrix that correspond to a Hamiltonian cycle $\mathcal{C}'$ of $\Pi(G)$, \textit{i.e.} for all edges $(u,v)$ of $\mathcal{C}'$, he sends $A_{u,v}$ as well as $\mathcal{C}'$.
	\end{itemize}
	\item The verifier checks that those decommitments are valid and correspond to what the provers have declared. He also checks that the timing constraint of the bit commitment is satisfied. This means that
	\begin{itemize}
		\item if $chall = 0$, the prover's opening $A$ must satisfy $\forall i,j \in [n], \ Y_{i,j} = A_{i,j} + (B_{i,j} * (M_{\Pi(G)})_{i,j})$.
		\item if $chall = 1$, the prover's opening $A$ must satisfy $\forall (u,v) \in \mathcal{C}', \ Y_{u,v} = A_{u,v} + B_{u,v}$.
	\end{itemize}
\end{enumerate} 
\end{minipage} 
} 

\subsection{Proof of security}
Our goal is to show that the above protocol is a relativistic zero-knowledge protocol for \HamiltonianCycle. In order to do this, we show the following
\begin{itemize}
	\item Completeness: If the prover and the verifier are honest then for any graph $G$ that has a Hamiltonian cycle, the verifier accepts with certainty.
	\item Soundness: If we take $Q = 64n!2^{3k}$, we have that for any cheating prover, $\forall G \notin \HamiltonianCycle$, the verifier accepts with probability at most $\frac{1}{2} + 2^{-k}$. With this parameter $Q$, the amount of bits sent during the protocol is $\log(Q)$ for each committed bit and is therefore $n^2 \log(Q) = O(kn^3\log(n))$ at each round.
	\item Perfect zero-knowledge: for any cheating verifier $V^*$, there exists a quantum poly-time simulator $\Sigma$ that can reproduce the cheating verifier's view of the protocol for any input $G \in \HamiltonianCycle$ and any auxiliary input $\rho$. More details about this zero-knowledge property can be found in the corresponding subsection. 
\end{itemize}
\subsubsection{Completeness.}
If both players are honest and $G$ contains a Hamiltonian cycle then the protocol always succeeds. Indeed, the original protocol from Blum has perfect completeness. Moreover, the $\FQ$ bit commitment always succeeds when done honestly.
\subsubsection{Soundness.}
The soundness can be reduced to the following 2-player game $G^{RZK-HAM}$.
\begin{itemize}
	\item $P_1$ receives a matrix $B \in \mathcal{M}_n^{\F_Q}$ where each element of $B$ is chosen uniformly at random in $\F_Q$. $P_2$ receives a random input bit $chall$.
	\item $P_1$ outputs a matrix $Y \in \mathcal{M}_n^{\F_Q}$. If $chall = 0$ then $P_2$ outputs a permutation $\Pi$ and a matrix $A \in \mathcal{M}_n^{\F_Q}$. If $chall = 1$ then $P_2$ outputs a cycle $\mathcal{C}'$ and $n$ strings $\{A'_{(u,v)}\}_{(u,v) \in \mathcal{C}'}$ in $\FQ$.
	\item If $chall = 0$, the two players win if $\forall i,j \in [n], \ Y_{i,j} = A_{i,j} + (B_{i,j} * (M_{\Pi(G)})_{i,j})$. If $chall = 1$, the two players win if for all edges $(u,v)$ of $\mathcal{C}'$, $Y_{u,v} = A_{u,v} + B_{u,v}$, which corresponds to revealing $1$ for each edge of the cycle $\mathcal{C}'$.
\end{itemize}
This game is $n!$-projective: once the permutation (or the cycle) is chosen, the winning output is fixed. In order to study this game, we study the game $G^{RZK-HAM}_{coup}$. We fix an input/output pair $(B,Y)$ for $P_1$ and we consider winning outputs for $P_2$ for both inputs. For $chall = 0$, we have a permutation $\Pi$ and a matrix $A \in \mathcal{M}_n^{\F_Q}$ which is a valid opening of $M_{\Pi(G)}$ meaning that 
\begin{align}\label{Eq:chall=0}
\forall (i,j), \ A_{i,j} = Y_{i,j} - B_{i,j} * (M_{\Pi(G)})_{i,j}.
\end{align}
For $chall = 1$, we have a cycle $\mathcal{C}'$ of $\{1,\dots,|V|\}$ as well as openings $A'_{u,v}$ for each $(u,v) \in \mathcal{C}'$. Because it is a winning output, the openings must satisfy
\begin{align} \label{Eq:chall=1}
\forall (u,v) \in \mathcal{C}', \ A'_{u,v} = Y_{u,v} - B_{u,v}.
\end{align}
If the graph $G$ (hence also $\Pi(G)$) does not contain a Hamiltonian cycle then there has to be an edge $(u,v)$ of $\mathcal{C}'$ such that $\left(M_{\Pi(G)}\right)_{u,v} = 0$. For this specific $(u,v)$, we combine Equations \ref{Eq:chall=0} and \ref{Eq:chall=1} and get:
$$A_{u,v} = Y_{u,v} \quad ; \quad A'_{u,v} = Y_{u,v} - B_{u,v}. $$
This implies that $A_{u,v} - A'_{u,v}  = B_{u,v}$ which happens with probability at most $\frac{1}{Q}$ from non-signaling. We therefore conclude that $\omega^*(G^{RZK - HAM}_{coup}) \le \frac{1}{Q}$. From there, we can apply Proposition \ref{Proposition:Quantum_coup_bound} and obtain 
$$ \omega^*(G^{RZK - HAM}) \le \frac{1}{2} + \left(\frac{64 n!}{Q}\right)^{1/3}. $$
If we take $Q = 64n!2^{3k}$ then the protocol has soundness $\frac{1}{2} + 2^{-k}$. The amount of bits sent during the protocol is $\log(Q)$ for each committed bit and is therefore $n^2 \log(Q) = O(kn^3\log(n))$, which shows that the protocol is efficient.

\subsection{Zero-knowledge property.}
In this section, we show that the above protocol is zero-knowledge. One of the main difficulties in proving zero-knowledge in the quantum setting arises when requiring the simulator to perform rewinding while preserving an auxiliary state. Here, there is no need for rewinding and the simulation can be done perfectly and quite simply. The simulator will simply simulate each round of the protocol from the first one to the last one. The reason of this simplicity is that in our bit commitment scheme, the verifier and the simulator are able, for any commitment, to reveal an arbitrary value of their choice. This is a rare feature because the prover shouldn't be able to do this to preserve the binding property. In our case, this asymmetry comes from the relativistic constraints imposed on the provers.

\paragraph{Zero-knowledge in the relativistic setting.}

From the provers' point of view, each of them receives a message and replies. We assume that a cheating verifier can totally bypass the timing constraints. We therefore consider one cheating verifier that interacts with both provers. Moreover, we allow the verifier to send a query to the second prover after receiving the answer from the first prover or vice-versa. All of this is meant to have a cheating verifier as strong as possible. Proving the zero-knowledge property in this setting will therefore be stronger in this model. Also, this will show the zero-knowledge property both for relativistic zero-knowledge and for the (very related) $2$-prover $1$-round multi-prover interactive proof model.

A cheating verifier $V^*$ is modeled by a polynomial-time uniform family of pairs of circuits $\{(V^*_1(n),V_2^*(n)\} $ where each $V_i^*(n)$ represents the verifier action towards prover $P_i$ on input size $n$. The verifier sends a query to each prover in respective classical registers $Q_1$ and $Q_2$ and gets responses in respective classical registers $R_1$ and $R_2$. The verifier also has access to private quantum register $\mathcal{V}$, which initially contains a quantum auxiliary state $\rho$.

Fix a cheating verifier $V^*$. For any message $B \in \mathcal{M}_n^{\FQ}$ sent from the verifier to $P_1$, the message from $P_1$ is a uniformly random matrix $Y = \mathcal{M}_n^{\FQ}$ while the message from $P_2$ consists of:
\begin{itemize}
	\item if $chall = 0$, a random permutation $\Pi$ and a matrix $A$ satisfying $Y = A + B * M_{\Pi(G)}$ where the multiplication is the entry-wise matrix multiplication.
	\item if $chall = 1$, a random cycle $\mathcal{C}'$ and a family of strings $\{A'_{u,v}\}_{(u,v) \in \mathcal{C}'}$ satisfying 
	$$ \forall (u,v) \in \mathcal{C}', \ Y_{u,v} = A'_{u,v} + B_{u,v}.$$
\end{itemize}
The verifier receives as a first message a random matrix $Y$ and as second message a random permutation (for chall = 0) or a cycle (for chall = 1) with a uniquely determined message $A$ or $A'$ that he can perfectly infer from the information available to him. Notice that in the soundness analysis, the prover doesn't know what message he has to send because of relativistic constraints which do not apply for the verifier (as we said, this only increases our claim on zero-knowledge). 

All of the above remains true for any strategy for the cheating verifier and with any auxiliary input, and even if the verifier queries a prover depending on the answer of the other prover. Moreover, simulating the interaction between $V^*$ and the provers can be done step by step following $V^*$'s actions, without any need for rewinding. This therefore shows the perfect zero-knowledge property of our scheme. In order to illustrate this, we present below a step by step simulation of the verifier's view in a more formal way than what we did above.

\paragraph{Step by step simulation of the verifier's view of the protocol.}

For a cheating verifier $V^*$, we construct a quantum poly-time simulator such that on any input $G \in \HamiltonianCycle$ and auxiliary input $\rho$, the simulator can recreate the verifier's view of the protocol perfectly. The simulator will use $V^*$ as a black box and will mimic the verifier's view of the protocol after each round. When considering the interaction between the verifier and the provers, we will always distinguish $2$ cases 
 \begin{enumerate}
 	\item The action of $V^*_2$ depends on the interaction with $P_1$.
 	\item The action of $V^*_1$ depends on the interaction with $P_2$.
 \end{enumerate}
 Note that both of these events cannot happen simultaneously. In the analysis below, we will consider case $1$ but the other one can be treated in the exact same way.

We first describe the different view for a cheating verifier $V^*$ and then show how to perform the simulation. Let $\sigma_i$ be the verifier's view at step $i$ of the protocol. 

\begin{itemize}
	\item At the beginning of the protocol, the verifier's view consists of $\sigma_0 := \rho_{\mathcal{V}}$.
	\item After the verifier's first message to $P_1$, the verifier's view is 
	$$\sigma_1 := V_1^*(\rho) = \sum_{B \in \mathcal{M}_n^{\F_Q}} p_{B} \altketbra{B}_{Q_1} \otimes \rho(B)_{\mathcal{V}}.$$
	\item After the first prover's answer, the shared state between the provers and the verifier is 
	$$\sigma_2 := \frac{1}{n!} \frac{1}{Q^{n^2}} \sum_{\Pi \in S_n}\sum_{A \in \mathcal{M}_n^{\F_Q}} \sum_{B \in \mathcal{M}_n^{\F_Q}} p_{B} \ \altketbra{Y(\Pi,A)}_{R_1} \otimes  \altketbra{B}_{Q_1} \otimes \rho(B)_{\mathcal{V}}.$$
	where $Y(\Pi,A) := A + B*\Pi(G)$ with $*$ being the entry wise matrix multiplication.
	\item Now, the verifier sends his challenge bit, which can depend on everything that happened before. His view becomes
	\begin{align*}
	 \sigma_3 := \frac{1}{n!} \frac{1}{Q^{n^2}} \sum_{\Pi \in S_n}\sum_{A \in \mathcal{M}_n^{\F_Q}}\sum_{B \in \mathcal{M}_n^{\F_Q}}\sum_{c \in \zo}   p_{B,c} \ \altketbra{Y(\Pi,A)}_{R_1} \otimes \altketbra{c}_{Q_2} \\
	  \otimes \altketbra{B}_{Q_1} \otimes \rho(B,c,Y(\Pi,A))_{\mathcal{V}}.
	 \end{align*}
	\item After the final message from the prover, the verifier's view becomes
	\begin{align*} \sigma_4 := \frac{1}{n!} \frac{1}{Q^{n^2}} \sum_{\Pi \in S_n}\sum_{A \in \mathcal{M}_n^{\F_Q}}\sum_{B \in \mathcal{M}_n^{\F_Q}} \altketbra{Y(\Pi,A)}_{R_1}  \otimes & \altketbra{B}_{Q_1}  \otimes \\ \Big(p_{B,0} \altketbra{0}_{Q_2} \otimes \altketbra{\Pi,A}_{R_2}  & \otimes \rho(B,0,Y(\Pi,A))  \\ +  p_{B,1} \altketbra{1}_{Q_2}  \otimes \altketbra{\Pi(\mathcal{C}),A_{\Pi(\mathcal{C})}}_{R_2} & \otimes \rho(B,1,Y(\Pi,A)) \Big). \end{align*}
\end{itemize}

Notice that we are interested here in the verifier's view on a 'Yes' instance, meaning that on challenge $'1'$ in register $Q_2$, the answer $\altketbra{\Pi(\mathcal{C}),A_{\Pi(\mathcal{C})}}$ satisfies
$$ \forall (i,j) \in \Pi(\mathcal{C}), \ Y_{i,j} = A_{i,j} + B_{i,j}.$$
meaning that the prover revealed the output bit '1' for entry $\Pi(G)_{i,j}$. Notice also that for a fixed cycle $\mathcal{C}$, the mapping $\Pi \rightarrow \Pi(\mathcal{C})$ is a bijection between the set of permutation and the set of cycles.  

We show now how to simulate the view of the verifier. The simulator can easily simulate $\sigma_0$ and $\sigma_1$ since he has a copy of $\rho$ and knows $V^*_1$. Notice that in $\sigma_2$, the message from the prover is a uniform random matrix because of the randomness $A$. Therefore, we have 
$$\sigma_2 = \frac{1}{Q^{n^2}} \sum_{Y \in \mathcal{M}_n^{\F_Q}} \sum_{B \in \mathcal{M}_n^{\F_Q}} p_{B} \ \altketbra{Y}_{R_1} \otimes  \altketbra{B}_{Q_1} \otimes \rho(B)_{\mathcal{V}}.$$
This can be easily created by the simulator by just tensoring the totally mixed state in register $R_1$ to $\sigma_1$. In order to construct $\sigma_3$, the simulator just applies $V^*_2$ to transform $\sigma_2$ into $\sigma_3$ as the cheating verifier would and gets exactly
$$ \sigma_3 =  \frac{1}{Q^{n^2}} \sum_{Y \in \mathcal{M}_n^{\F_Q}}\sum_{B \in \mathcal{M}_n^{\F_Q}}\sum_{c \in \zo}   p_{B,c} \ \altketbra{Y}_{R_1} \otimes \altketbra{c}_{Q_2} \otimes \altketbra{B}_{Q_1} \otimes \rho(B,c,Y)_{\mathcal{V}}.$$
Finally, in order to construct $\sigma_4$, the simulator does the following
\begin{itemize}
	\item conditioned on $c = 0$ in register $Q_2$, the simulator picks a random permutation $\Pi$ and puts $\kb{\Pi,A(\Pi,B,Y)}$ in register $R_2$ where $A(\Pi,B,Y) := Y - B * \Pi(G)$, with $*$ being the entry wise matrix multiplication.
	\item conditioned on $c=1$ in register $Q_2$, the simulator picks a random cycle $\mathcal{C}'$ and outputs $\kb{\mathcal{C}',A'(\mathcal{C}',B,Y)}$ such that for all $(i,j) \in \mathcal{C}'$, it holds that $ A'(\mathcal{C}',B,Y)_{i,j} := Y_{i,j} - B_{i,j}$.
\end{itemize}
The state constructed by the simulator is therefore
\begin{align*}
\frac{1}{Q^{n^2}}\sum_{A \in \mathcal{M}_n^{\F_Q}}\sum_{B \in \mathcal{M}_n^{\F_Q}} \altketbra{Y}_{R_1}  \otimes \altketbra{B}_{Q_1} \otimes \\
(p_{B,0} \kb{0}_{Q_2} \otimes \frac{1}{n!} \sum_{\Pi \in S_n} \kb{\Pi,A(\Pi,B,Y)}_{R_2} \otimes \rho_{B,0,Y} \\
+ p_{b,1} \kb{1}_{Q_2} \otimes \frac{1}{(n-1)!} \sum_{\mathcal{C}' \in \Gamma_n} \kb{\mathcal{C}',A'(\mathcal{C}',B,Y)}_{R_2} \otimes \rho_{B,1,Y}).
\end{align*}
By simple changes of variables, we can see that the above state is actually exactly equal to $\sigma_4$. Therefore, we succeeded in the simulation and we can conclude that our protocol is perfectly zero-knowledge against quantum adversaries.
\section*{Acknowledgements}
The authors were partially supported by ANR DEREC $<$ANR-16-CE39-0001-01$>$.
\bibliography{paper}
\bibliographystyle{alpha}
\end{document}